\documentclass[APA,LATO1COL]{WileyNJD-v2}

\usepackage{booktabs}
\usepackage[utf8]{inputenc}
\usepackage{amsmath, physics}
\usepackage{amssymb}
\usepackage{amsfonts}
\usepackage{amscd}
\usepackage{caption}
\usepackage{subcaption}
\usepackage{latexsym}
\usepackage{wasysym}
\usepackage{graphicx}
\usepackage[all,cmtip]{xy}
\usepackage{bbold}
\usepackage{mathtools}
\usepackage{multirow}
\usepackage{xspace}
\usepackage{multirow}
\usepackage{amsfonts}
\usepackage{amssymb}
\usepackage{amsthm}
\usepackage{enumitem}
\usepackage{mathtools,bbm}
\usepackage{xpatch,multicol,array}
\usepackage{lipsum}
\usepackage{mathtools}
\usepackage{cuted}
\usepackage{mathtools,bbm}
\usepackage{tabularx}
\graphicspath{{./images/}{./fig}}

\newcommand{\N}{\mathbb{N}}

\newcommand{\R}{\mathbb{R}}

\DeclareSymbolFont{bbold}{U}{bbold}{m}{n}
\DeclareSymbolFontAlphabet{\mathbbold}{bbold}
\newcommand{\vect}[1]{\mathbbold{#1}}

\newcommand{\setdef}[2]{\left\{#1 \; | \; #2\right\}}

\usepackage{algpseudocode,algorithm,algorithmicx}
\usepackage{algpseudocode}

\newcommand{\subscr}[2]{#1_{\textup{#2}}}

\newcommand{\Sd}{\subscr{\mathbb{S}}{d}}

\newcommand{\Sdpp}{\subscr{\mathbb{S}}{d}^{\operatorname{++}}}
\newcommand{\vecc}{\textbf{\textup{vec}}}

\DeclareMathOperator\arctanh{arctanh}

\def\limn{\lim_{n \rightarrow \infty}}

\DeclareSymbolFont{bbold}{U}{bbold}{m}{n}
\DeclareSymbolFontAlphabet{\mathbbold}{bbold}
\usepackage{algpseudocode,algorithm,algorithmicx}
\usepackage{algpseudocode}

\def\limn{\lim_{n \rightarrow \infty}}

\makeatletter
\newcommand{\multiline}[1]{%
  \begin{tabularx}{\dimexpr\linewidth-\ALG@thistlm}[t]{@{}X@{}}
    #1
  \end{tabularx}
}
\makeatother

\articletype{Article Type}%
\received{26 April 2016}
\revised{6 June 2016}
\accepted{6 June 2016}

\begin{document}

\title{Family-wise error rate control in Gaussian graphical model selection via Distributionally Robust Optimization
}

\author[1]{Chau Tran}

\author[2]{Pedro Cisneros-Velarde}

\author[1,4]{Sang-Yun Oh*}

\author[3]{Alexander Petersen}

\authormark{Tran \textsc{et al}}

\address[1]{\orgdiv{Department of Statistics and Applied Probability}, \orgname{University of California Santa Barbara}, \orgaddress{\state{CA}, \country{US}}}

\address[2]{\orgdiv{Department of Computer Science}, \orgname{University of Illinois at Urbana-Champaign}, \orgaddress{\state{IL}, \country{US}}}

\address[3]{\orgdiv{Department of Statistics}, \orgname{Brigham Young University}, \orgaddress{\state{UT}, \country{US}}}

\address[4]{\orgdiv{Scientific Data Division}, \orgname{Lawrence Berkeley National Laboratory}, \orgaddress{\state{CA}, \country{US}}}

\corres{*Sang-Yun Oh\\ \email{syoh@ucsb.edu}}

\presentaddress{Department of Statistics \& Applied Probability\\
University of California\\
Santa Barbara, CA 93106-3110}

\abstract[Summary]{Recently, a special case of precision matrix estimation based on a distributionally robust optimization (DRO) framework has been shown to be equivalent to the graphical lasso. 
From this formulation, a method for choosing the regularization term, i.e., for graphical model selection, was proposed. 
In this work, we establish a theoretical connection between the confidence level of graphical model selection via the DRO formulation and the asymptotic family-wise error rate of estimating false edges. Simulation experiments and real data analyses illustrate the utility of the asymptotic family-wise error rate control behavior even in finite samples.}


\maketitle

\section{Introduction}\label{sec1}

The estimation of the precision matrix $\Omega=\Sigma^{-1}$ of a Gaussian random vector $X \in \R^d$ with covariance matrix $\Sigma$ is a problem that has received much attention in statistics and machine learning \citep{dempster1972, Drton-Perlman, MY-LY:07, Drton2017}. The matrix $\Omega$ characterizes the \emph{conditional dependency} structure between variables.  If a random variable $X$ follows a normal distribution, $\Omega_{jk} = 0$ if and only if the $j$-th and $k$-th variables of $X$ are conditionally independent given the rest \citep{lauritzen1996}. 

Naturally, an $\ell_1$-regularized maximum likelihood approach that introduces sparsity in the estimation of $\Omega$ was proposed by \citep{MY-LY:07}. The approach will be referred to by the name of a well-known computational algorithm, graphical lasso \citep{JF-TH-RT:07}. The resulting sparsity pattern from graphical lasso can be then used to construct a graphical model, $G=(V,E)$, where $V$ is the set of nodes for each of the $d$-variables, and $E$ is the set of undirected edges: each edge $(i,j)$ represents a non-zero element for $i$-th and $j$-th variables in $\Omega$. Graphical lasso subsequently spurred significant research effort in methodological development as well as application domains \citep{guillot2015,HUANG2010935,Krumsiek2011}. As with most other learning methods, the performance of graphical lasso depends on a user-specified tuning parameter; however, tuning the sparsity-inducing regularization parameter of graphical lasso --- also called \emph{graphical model selection} --- is often challenging for various reasons.

In practice, procedures such as cross-validation (CV) and Bayes information criterion (BIC) minimization are often used to tune graphical lasso; however, they tend to overfit in simulation experiments \citep{hastie2009,liu2010}. Furthermore, CV and BIC minimization are computationally demanding because they search over a grid of candidate parameters. Moreover, asymptotic properties in the literature are often not beneficial for regularization parameter tuning in finite sample regimes. As a result, using graphical lasso in real applications is often met with significant computational and statistical subtleties, and, hence, practitioners sometimes resort to manual tuning in order to obtain an estimate of $\Omega$ with a targeted number of non-zeros.

Recently, \cite{cisneros20a} has formulated the precision matrix estimation problem using the distributionally robust optimization (DRO) framework \citep{VAN-DK-PME:18, JB-NS:19}. The authors establish the correspondence between the radius of the \emph{ambiguity set} in the DRO framework --- which measures the uncertainity around the empirical measure (see more below) --- and the regularization parameter of graphical lasso estimator. The authors leveraged this connection to propose a \emph{robust selection} (RobSel) algorithm that, given a confidence level $1-\alpha$, determines the corresponding regularization parameter for graphical lasso.

Our work theoretically relates the RobSel error tolerance $\alpha$ to the asymptotic family-wise error rate (FWER) for estimating any false positive non-zero in $\Omega$. The practical significance of our work is that graphical lasso regularization can be chosen according to a user specified FWER level. We illustrate the theoretical result in simulation and compare the similarity between RobSel chosen graphs and graphs estimated by a hypothesis testing-based procedure for graphical model selection. We confirm that choosing graphical lasso regularization parameter with RobSel can still yield a consistent family-wise error rate characteristic in finite samples.

\section{DRO formulation and family-wise error rate of graphical lasso}\label{sec2}

Distributionally robust optimization (DRO) as an estimation framework seeks parameters that minimize the worst expected risk over the uncertainty set of distributions (often called \emph{ambiguity set} in DRO terminology). Readers are referred to a review article by \cite{kuhn2019} for an overview of the DRO. Leveraging the DRO framework, \cite{cisneros20a} showed that for a fixed $\rho \ge 1$ and $p \in [1,\infty]$, their DRO formulation of regularized inverse covariance estimation is equivalent to the following expression:
\begin{align}
\label{eq1}
\min_{K\in\Sdpp}\left\{\trace(KA_n)-\log|K|+\delta^{1/\rho}\norm{\vecc(K)}_p\right\},
\end{align}
where $\Sdpp$ denotes the set of $d \times d$ positive definite matrices, and $\delta$ is the radius of ambiguity set, which is constructed as a ball in the Wasserstein space of distributions, centered at the empirical measure of the data. Note that graphical lasso objective function is a special case of \eqref{eq1} when $p=1$ and $\rho=1$. Constants $p$ and $\rho$ specify the Wasserstein distance metric between two probability distributions \cite[see][for details]{cisneros20a}. Remarkably, the regularization parameter of graphical lasso corresponds to the ambiguity set radius $\delta$ despite the differing premise between DRO and maximum likelihood estimator. Intuitively, an increase in ambiguity set radius $\delta$ (i.e., an increased robustness in DRO) corresponds to an increased amount of regularization in graphical lasso (which results in conservative selection of non-zeros).

Using the \emph{Robust Wasserstein Profile (RWP) function} $R_n$ introduced by~\cite{blanchet_kang_murthy_2019}, \cite{cisneros20a} derived the RWP function for graphical lasso, $R_{n}(K) = \norm{\vecc(A_n-K)}_\infty$, and characterized its asymptotic distribution. The distribution is used to determine $\delta$ \citep[equivalently, the regularization parameter $\lambda$ in graphical lasso][]{JF-TH-RT:07} given the user specified error tolerance level $\alpha$:
\begin{align}
\label{oii2}
\lambda = \delta := 
\inf\setdef{\delta>0}{\mathbb{P}_0(R_n(\Omega)\leq \delta)}
=\inf\setdef{\delta>0}{\mathbb{P}_0(\norm{\vecc(A_n-\Sigma)}_\infty\leq \delta) \geq
1-\alpha},
\end{align}
where $\mathbb{P}_0$ denotes the true underlying distribution of the data.   
This graphical model selection procedure is called \emph{RobSel} in~\citep{cisneros20a}. 
Then, by Corollary 3.3 of \cite{cisneros20a}, $n^{1/2}\delta$ tends to $1-\alpha$ quantile of $R_n$, $r_{1-\alpha}$, and the corresponding $\delta$ can be determined from an order statistic in finite sample. The asymptotic result also motivates the approximation of the RWP function through a bootstrap procedure in Algorithm \ref{alg:robsel} to determine the regularization parameter $\lambda$, given significance level $\alpha$.

\subsection{Family-wise error rate control with RobSel}

\begin{algorithm}[Ht!]
\caption{RobSel algorithm for estimation of the regularization parameter $\lambda$~\citep{cisneros20a} \label{alg:robsel}}
\begin{algorithmic}
    \State Input: $n$ observations, $X_{1},\ldots, X_{n}$.
    \State Set parameters $ \alpha \in (0,1)$ and $B \in \N$.
    \State Compute empirical covariance $A_n$.
    \For{$b=1,...,B$}
        \State Obtain a bootstrap sample $X_{1b}^*,\ldots, X_{nb}^*$ by sampling uniformly and with replacement from the data
        \State Compute empirical covariance $A^*_{n,b}$ from the bootstrap sample.
        \State $R^*_{n,b} \gets \norm{A^*_{n,b}-A_n}_\infty$
    \EndFor
    \State Set $\lambda$ to be the bootstrap order statistic $R^*_{n, ((B+1)(1-\alpha))}$.
\end{algorithmic}
\end{algorithm}
In this next section, we provide results for the interpretation of $\alpha$ and its relation to type I error control in graphical model selection. Recall that equation \eqref{eq1} shows that the DRO estimator is equivalent to the $\ell_1$-penalized estimator in graphical lasso, which produces a sparse estimator of $\Omega,$ denoted $\hat{\Omega}^\delta$. Given equation \eqref{oii2}, a natural question is how to interpret error tolerance $\alpha$, which was not addressed in \cite{cisneros20a}. The following result directly connects parameter $\alpha$ in RobSel and the asymptotic FWER of the corresponding obtained estimator.
\begin{theorem}[FWER of graphical lasso]
\label{thm-FWER} 
Let $\Xi = \{(i,j): \Omega_{ij} = 0\}$ be the indices corresponding to zero entries of $\Omega.$  For a fixed $\alpha,$ let $\delta$ satisfy \eqref{oii2} and let $\hat{\Omega}^\delta$ be the unique solution to optimization problem~\eqref{eq1} with $\rho=1$. Then
\begin{align}
\label{FWER}
\limn \mathbb{P}(\hat{\Omega}^\delta_{ij} \neq 0 \textrm{ for some } (i,j) \in \Xi) \leq \alpha. 
\end{align}
\end{theorem}

\begin{proof}
In this proof, let $\mathbb{S}_{\mathrm{d}}$ be the set of $d\times d$ symmetric matrices. 
Recall that $n^{1/2}\delta \rightarrow r_{1 - \alpha},$ where $r_{1 - \alpha}$ is the $1 - \alpha$ quantile of the distribution in Corollary 3.3 and Remark 3.5 of \cite{cisneros20a}. Then, by Theorem 1 of \cite{MY-LY:07}, we have that $n^{1/2}(\hat{\Omega}^\delta - \Omega)$ converges in distribution to $U^*$, the minimizer of
$$
\arg \min_{U = U'}\,\, \trace(U\Sigma U \Sigma) + \trace(UH) + r_{1 - \alpha}\sum_{i \neq j}\left\{u_{ij}\mathrm{sign}(\Omega_{ij})\mathbf{1}(\Omega_{ij} \neq 0) + |u_{ij}|\mathbf{1}(\Omega_{ij} = 0)\right\},
$$
where $H \in \mathbb{S}_{\mathrm{d}}$ is a matrix of jointly Gaussian random variables with zero mean such that $\operatorname{Cov}\left(h_{i j}, h_{k \ell}\right)=E\left[x_{i} x_{j} x_{k} x_{\ell}\right]-\Sigma_{i j} \Sigma_{k \ell}$. By the convex nature of the above optimization problem, using the first optimality criterion using subdifferentials~\citep[Corollary~2.7]{FHC-YSL-RJS-PRW:98},  it follows that there exists some $Z \in \Sd$ satisfying
$$
Z_{ij} = \left\{\begin{array}{ll} 0, & i = j, \\ \mathrm{sign}(\Omega_{ij}), & i \neq j, \Omega_{ij} \neq 0, \\ \mathrm{sign}(u_{ij}), &i \neq j, \Omega_{ij} = 0, u_{ij} \neq 0,\\ \in [-1,1], & i \neq j, \Omega_{ij} = u_{ij} = 0. \end{array} \right.
$$
for which $H + 2\Sigma U^* \Sigma + r_{1 - \alpha}Z = 0.$ Letting $\otimes$ denote the matrix Kronecker product and $\Gamma = \Sigma \otimes \Sigma$, it follows that
$$
\vecc(U^*) = -\frac{1}{2}\Gamma^{-1}\left\{\vecc(H) + r_{1 - \alpha}\vecc(Z)\right\}.
$$
Finally, let $\hat{\Omega}_{\Xi}^\delta$ denote the vector of elements of $\hat{\Omega}^\delta$ whose indices are in $\Xi$, $\Omega_{\Xi}$ denote the vector of elements of $\Omega$ whose indices are in $\Xi$ (so it is the zero vector), and 
$U^*_{\Xi}$ denote the vector of elements of $U^*$ whose indices are in $\Xi$. Then one concludes that
\[
\begin{split}
    \lim_{n \rightarrow \infty} \mathbb{P}(\hat{\Omega}^\delta_{ij} \neq 0 \textrm{ for some } (i,j) \in \Xi) &= \limn \mathbb{P}(\sqrt{n}(\hat{\Omega}^\delta_{\Xi} - \Omega_{\Xi}) \neq 0) = \mathbb{P}(U^*_{\Xi} \neq 0) \\
    &\leq \mathbb{P}(U^* \neq 0) = 1 - \mathbb{P}(H \neq -r_{1 - \alpha}Z) \leq 1 - \mathbb{P}(\norm{\vecc(H)}_\infty \leq r_{1 - \alpha}) \\
    &=\alpha.
\end{split}
\]
\end{proof}

Using the estimated regularization parameter $\lambda(\alpha)$ from RobSel for graphical lasso, Theorem \ref{thm-FWER} states that the asymptotic probability that the estimated graph includes a false edge (false non-zero estimated in $\hat\Omega^\delta$) is bounded by $\alpha$. This interpretation is equivalent to having the FWER bounded by $\alpha$ in hypothesis testing-based graphical model selection in \cite{Drton-Perlman}. As a result, Theorem \ref{thm-FWER} implies that RobSel can also serve as a tool for controlling graphical lasso's FWER at some chosen significance level $\alpha$ with similar to using a hypothesis testing-based graphical model selection.

Concretely, testing $d(d-1)/2$ null hypotheses that each pairwise partial correlation is zero can serve as an alternative way to construct a graphical model, where the partial correlation between variables $i$ and $j$ is defined as $\rho_{i j{\boldmath\cdot}\text{rest}}=-\omega_{i j}/\sqrt{\omega_{i i} \omega_{j j}}$ and $i,j=1,2,\dots,d$. The unadjusted $p$-value $\pi_{ij}$ for each null hypothesis is be obtained by 
\begin{equation}
    \label{unadjusted-p-value}
    \pi_{ij} = 2[1-\Phi(\sqrt{n - d - 1}\cdot \abs{z_{ij{\boldmath\cdot}\text{rest}}})],
\end{equation}
where $\Phi$ is the CDF of standard normal distribution, $z_{ij{\boldmath\cdot}\text{rest}} = \arctanh(r_{ij{\boldmath\cdot}\text{rest}})$ is the Fisher's $z$-transformed sample partial correlation $r_{ij{\boldmath\cdot}\text{rest}}$ for population partial correlation $\rho_{i j{\boldmath\cdot}\text{rest}}$. To account for multiple comparison, a p-value correction is needed to achieve a desired FWER characteristic. One of the multiple testing correction methods given in \cite{Drton-Perlman} controls the FWER based on Holm's approach for $p$-value adjustment: 
\begin{align}
    \pi_{a\uparrow}^{\text{Holm}} = \max_{b=1,...,a}\left[\min\left\{\left(\binom{d}{2} - b + 1\right)\pi_{b\uparrow}, 1 \right\}\right]\text{, for } 1 \le a \le \binom{d}{2}.
\end{align}
where $\pi_{1\uparrow} \le \pi_{2\uparrow} \le ... \le \pi_{d(d-1)/2\uparrow}$ are the ordered $p$-values from \eqref{unadjusted-p-value}. This approach will be referred to as the Holm-corrected testing method for graphical model selection in our numerical experiments. Other multiple testing correction approaches discussed in \cite{Drton-Perlman} include Bonferroni and Šidák adjustments. For the remainder of our work, we compare RobSel with the Holm-corrected testing method for its simplicity (compared to the Šidák-based approach) and better power characteristic (compared to the Bonferroni-based approach). We emphasize that the distinct advantage of graphical lasso is that it can can perform model selection and parameter estimation of $\Omega$ simultaneously, whereas any testing-based approach can only identify the zeros/non-zero locations of $\Omega$.


\section{Numerical simulations and real data analysis}\label{sec3}

In this section, analyses of simulated and real data illustrate the usefulness of RobSel's asymptotic FWER property in finite sample and compare to the Holm-based multiple testing approach for Gaussian graphical model selection. Furthermore, RobSel is used to analyze real datasets from genomics. 

To carry out our numerical experiments, we used packages \texttt{CVglasso} for cross validation, \texttt{qgraph} for the extended Bayesian information criterion, and \texttt{robsel} for Robust Selection. These packages are from CRAN, and they use package \texttt{glasso} to estimate the sparse inverse covariance matrix. Robust Selection algorithm is also available as a Python package, \texttt{robust-selection}, at \url{https://pypi.org/project/robust-selection/}. Both Python and R packages are also available at \url{https://github.com/dddlab/robust-selection}, and the codes to reproduce the numerical results is available at \url{https://github.com/cbtran/robsel-reproducible}.

\subsection{Simulation experiments}\label{sec:setting}

\begin{figure}[Ht!]
\centering
     \begin{subfigure}[b]{0.33\textwidth}
         \centering
         \includegraphics[width=\textwidth]{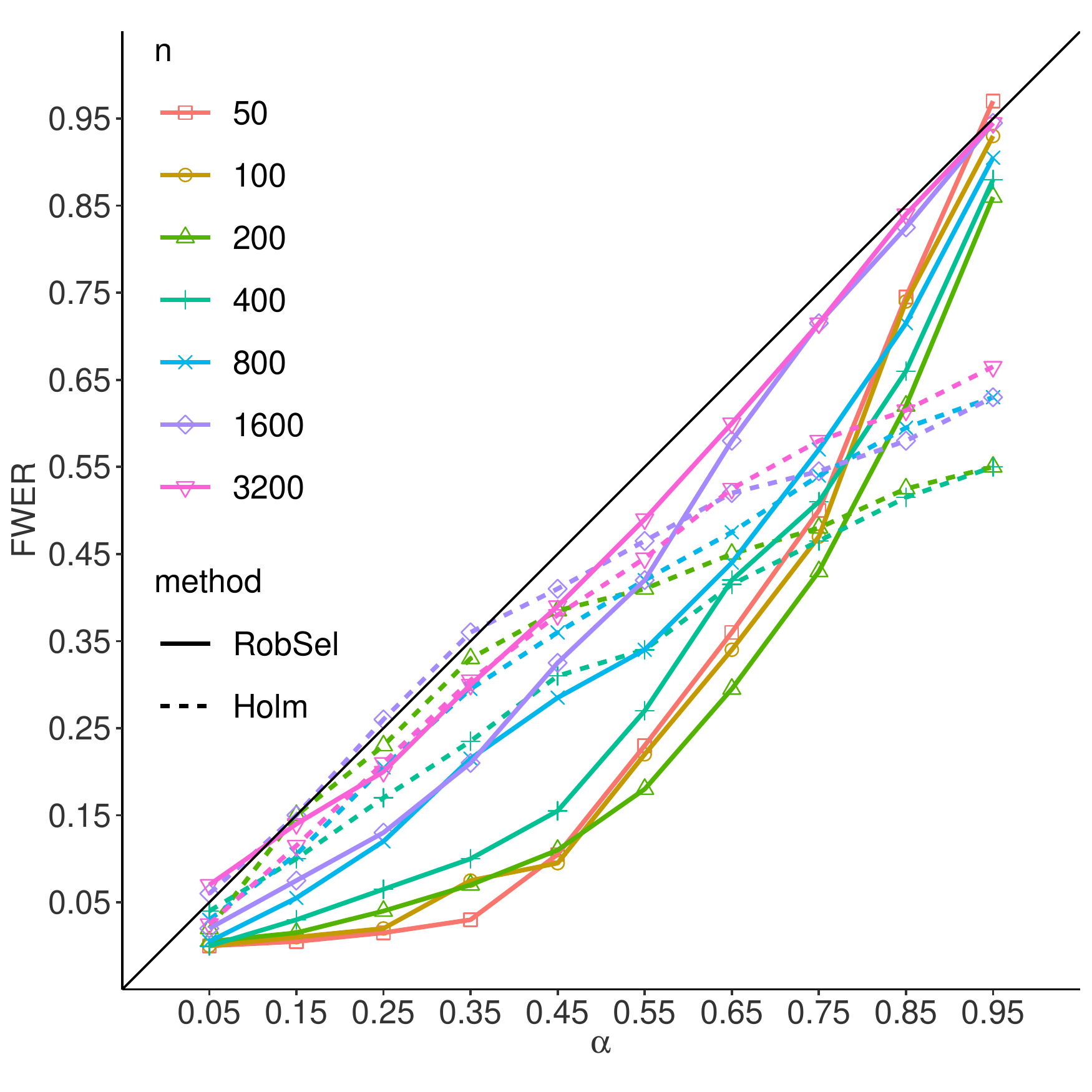}
     \end{subfigure}
     \begin{subfigure}[b]{0.33\textwidth}
         \centering
         \includegraphics[width=\textwidth]{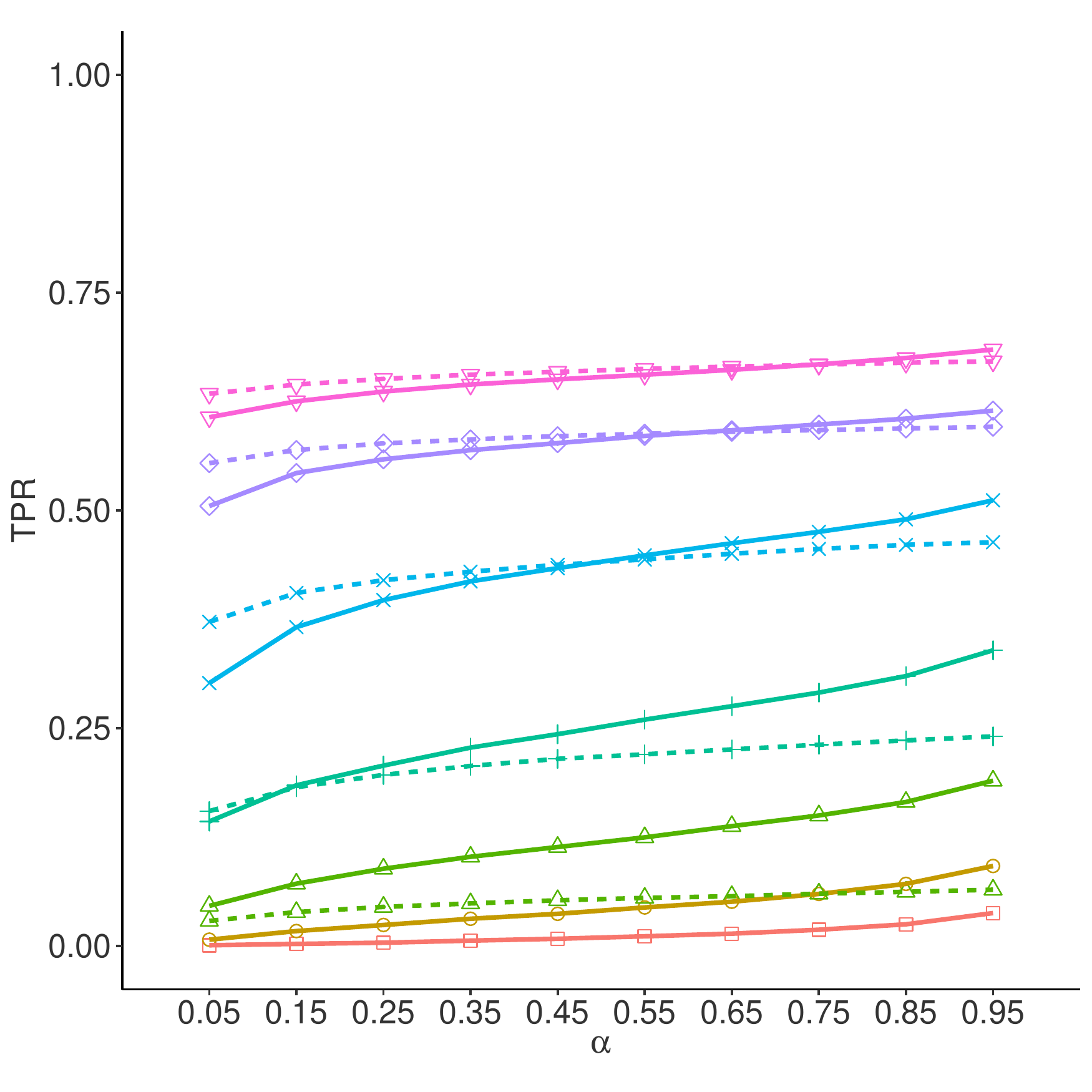}
     \end{subfigure}
     \begin{subfigure}[b]{0.33\textwidth}
         \centering
         \includegraphics[width=\linewidth]{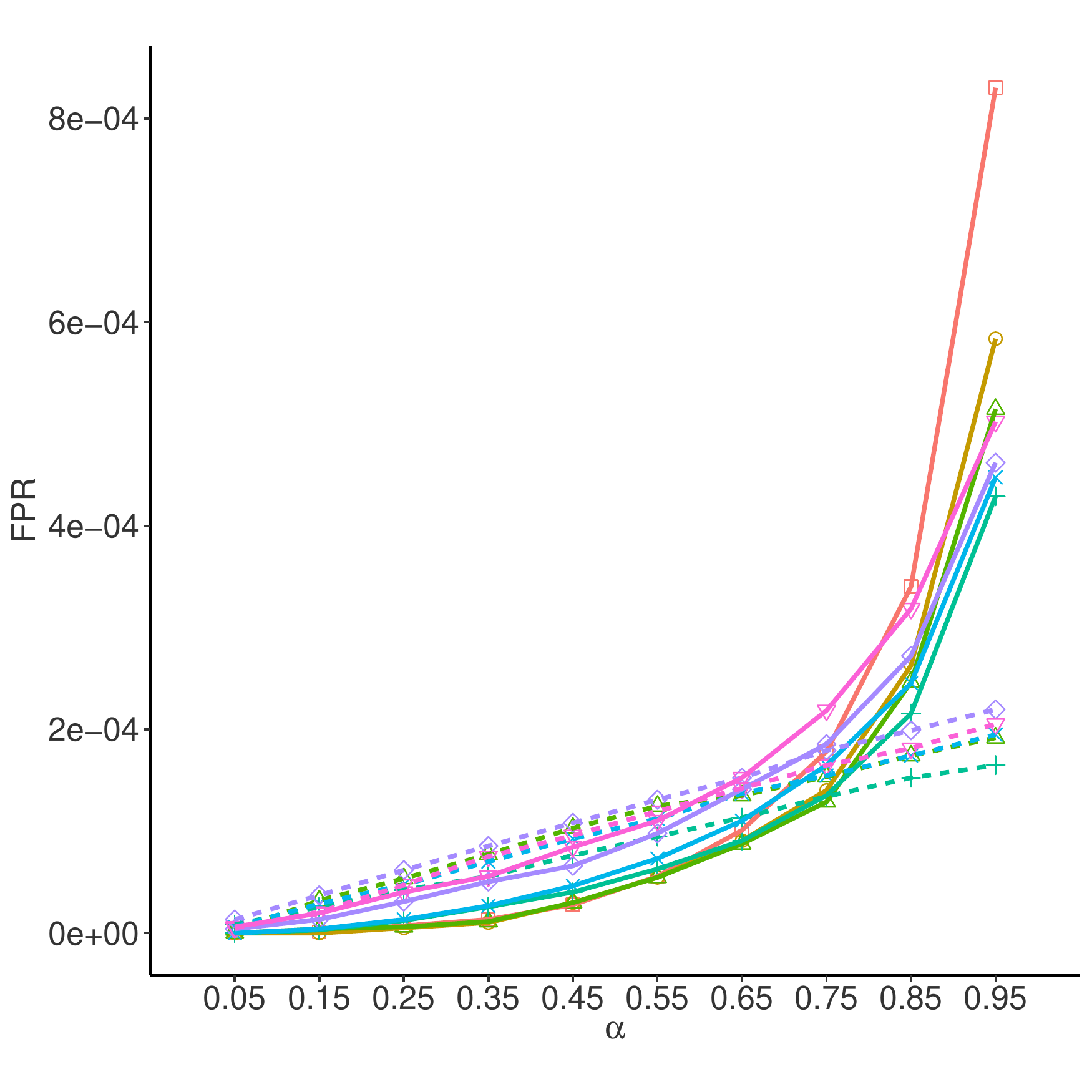}
    \end{subfigure}
    \begin{subfigure}[b]{0.33\textwidth}
         \centering
         \includegraphics[width=\linewidth]{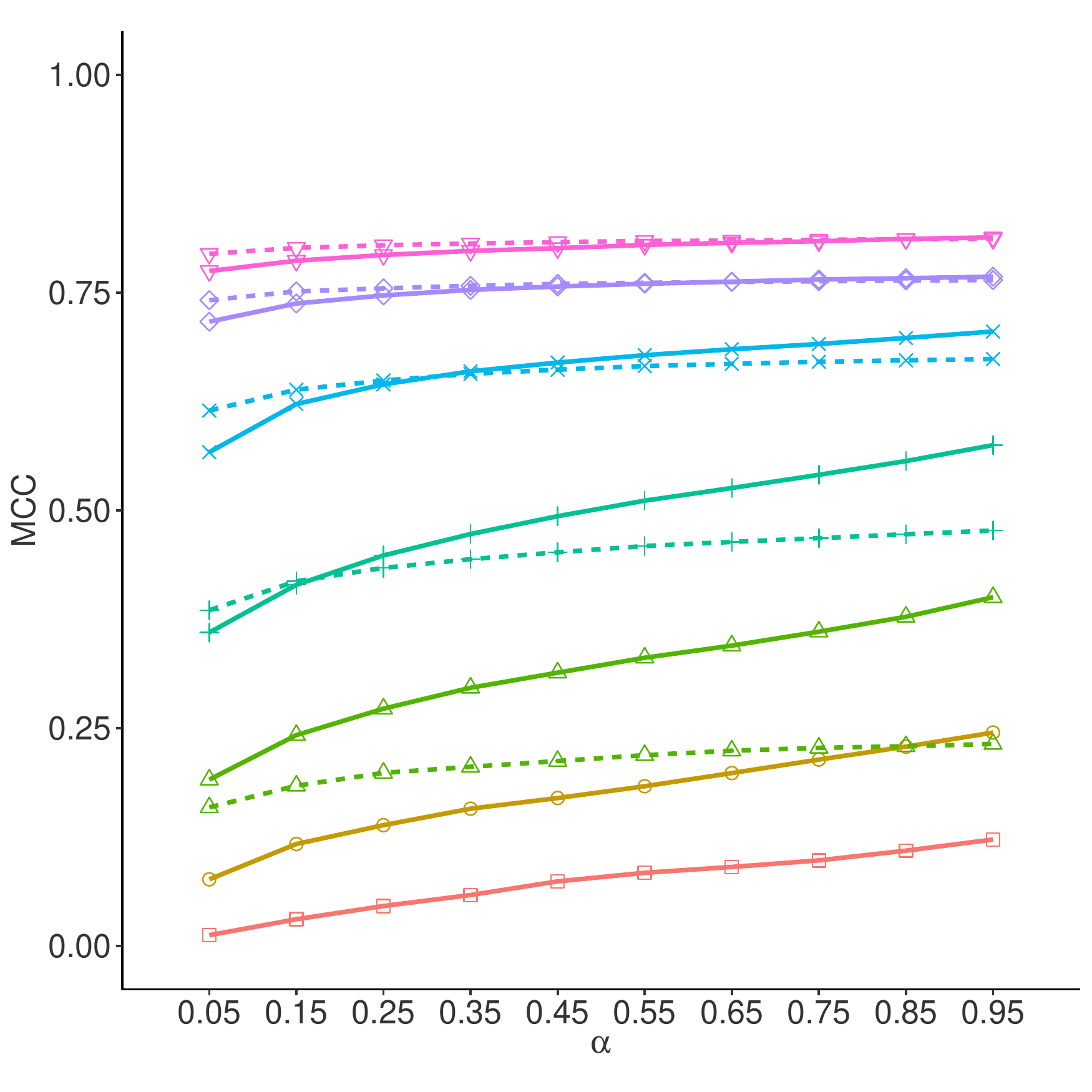}
    \end{subfigure}
    \begin{subfigure}[b]{0.33\textwidth}
         \centering
         \includegraphics[width=\linewidth]{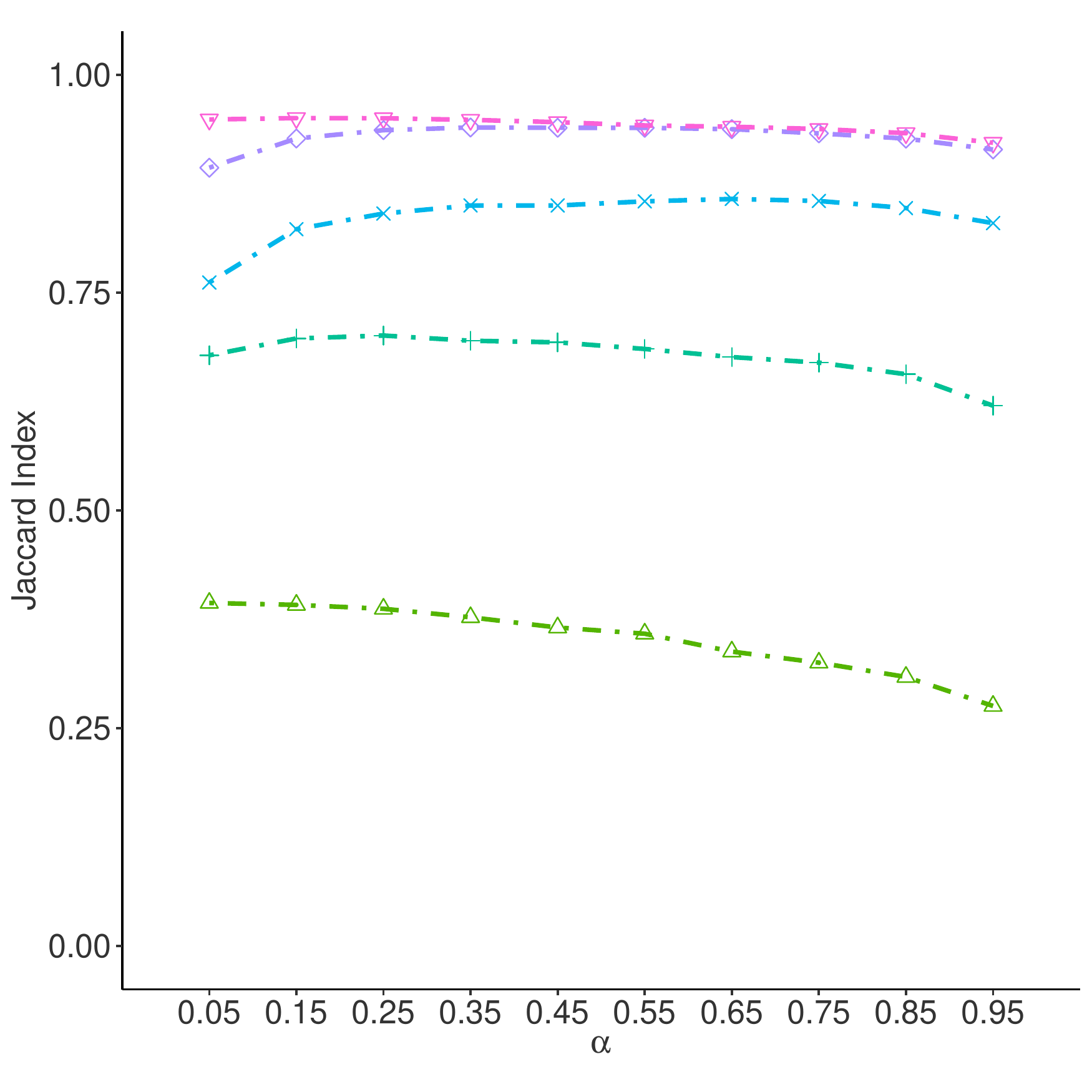}
     \end{subfigure}
\caption{
Observed family-wise error rate (top-left), True Positive Rate (top-middle), False Positive Rate (top-right), Matthews Correlation Coefficient (bottom-left), and Jaccard index of similarity (bottom-right) evaluated from graphs estimated with RobSel with graphical lasso and Holm-based multiple testing method. Note that Holm-based method is not applicable when $n\leq d=100$. All traces represent average quantities over 200 datasets. 
}
\label{fig:FWER_GLasso}
\end{figure}

In applications, the finite sample behavior the FWER characteristic whose asymptotic properties are given in Theorem \ref{thm-FWER} is of practical interest. In this section, simulation studies are used to verify FWER of graph reconstruction when using RobSel with graphical lasso. Also, the FWER of a testing-based graphical model selection from \cite{Drton-Perlman} is given as a comparison. (Readers are referred to \cite{cisneros20a} for comparison to cross-validation procedure.)

The true precision matrix $\Omega\in\Sdpp$ used to generate the simulated data has been constructed as follows. First, generate an adjacency matrix of an undirected {Erd\H{o}s-Renyi} graph with equal edge probability of 0.02 discarding any self-loops. Then, the weight of each edge (the magnitude of the non-zero element) is sampled uniformly between $[0.5, 1]$, and the sign of each non-zero element is set to be positive or negative with equal probability of 0.5. The resulting matrix is made diagonally dominant by following a procedure described in~\citep{Peng2009}, which ensures that the resulting matrix $\Omega$ is positive definite with ones on the diagonal. Finally, the diagonal entries of $\Omega$ are resampled uniformly between $[1, 1.5]$. Throughout this numerical study section, one randomly generated instance of sparse matrix $\Omega$ with $d=100$ variables is fixed. Using this $\Omega$, a total of $N=200$ datasets for each sample size $n\in\{50,100,200,400,800,1600,3200\}$ were generated independently from a multivariate zero-mean Gaussian distribution, i.e., $\mathcal{N}(\vect{0}_d,\Omega^{-1})$.

To evaluate the selected models, family-wise error rate (FWER), true positive rate (TPR), false positive rate (FPR),  Matthews correlation coefficient (MCC), and Jaccard index were used as performance metrics. These metrics are derived from elements in the confusion matrix, true positives (TP), true negatives (TN), false positives (FP) and false negatives (FN), where a positive indicates an estimated presence of an edge (two non-zero entries in $\Omega$). In this setting, family-wise error rate is the probability of any false edge detection: $FWER = \mathbf{1}( FP > 0$).  
True positive rate is the proportion of edges in true graph $G$ that are correctly identified in the estimated graph: $TPR=\frac{TP}{TP+FN}$. False positive rate is the proportion of nonedges in true graph $G$ that are incorrectly identified as edges in the estimated graph: $FPR=\frac{FP}{FP+TN}$. Matthews correlation coefficient summarizes all count in confusion matrix to measure quality of graph recovery performance: $MCC = \frac{TP\cdot TN - FP \cdot FN}{(TP + FP)(TP+FN)(TN+FP)(TN+FN)}$.
Jaccard index measure the similarity between two edge sets $E_A$ and $E_B$: $J(E_A,E_B)=\frac{|E_A \cap E_B|}{|E_A \cup E_B|}$, and, by convention, Jaccard index of two empty sets is defined to be one, i.e., $J(\emptyset,\emptyset) = 1$.

Figure \ref{fig:FWER_GLasso} shows the FWER, TPR, FPR, MCC, and Jaccard index of the estimated graphs from both Holm's multiple testing method and the graphical lasso with RobSel criterion. TPR increases as sample size increases; however, for each sample size, both method have similar TPR, but RobSel appears to be more conservative at small significant levels since it tends to have smaller TPR and FWER. For larger $\alpha$, RobSel is less conservative with higher TPR while its FWER still bounded by $\alpha$. Figure \ref{fig:FWER_GLasso} also show the average Jaccard index from 200 simulations at 5 different sample size and 10 different levels $\alpha$. It can be seen that Jaccard index increases as sample size increases indicating the estimated graphs from both RobSel and Holm-based multiple testing method become increasingly similar.

Figure \ref{fig:estimated_graph} illustrates a striking similarity between graphical lasso tuned with RobSel and testing-based graphs for large $n$. Most edges appear in both graphs and both graphs do not contain any false positive edge owing to the stringent significance level. On the other hands, graphical lasso tuned with cross-validation have many false positive edges. These qualitative observations were typical in our numerical simulations when data were generated from multivariate normal distributions across a wide range of sample sizes we considered.

\begin{figure}[Ht!]
     \centering
     \begin{subfigure}[b]{0.24\textwidth}
         \centering
         \includegraphics[width=\textwidth]{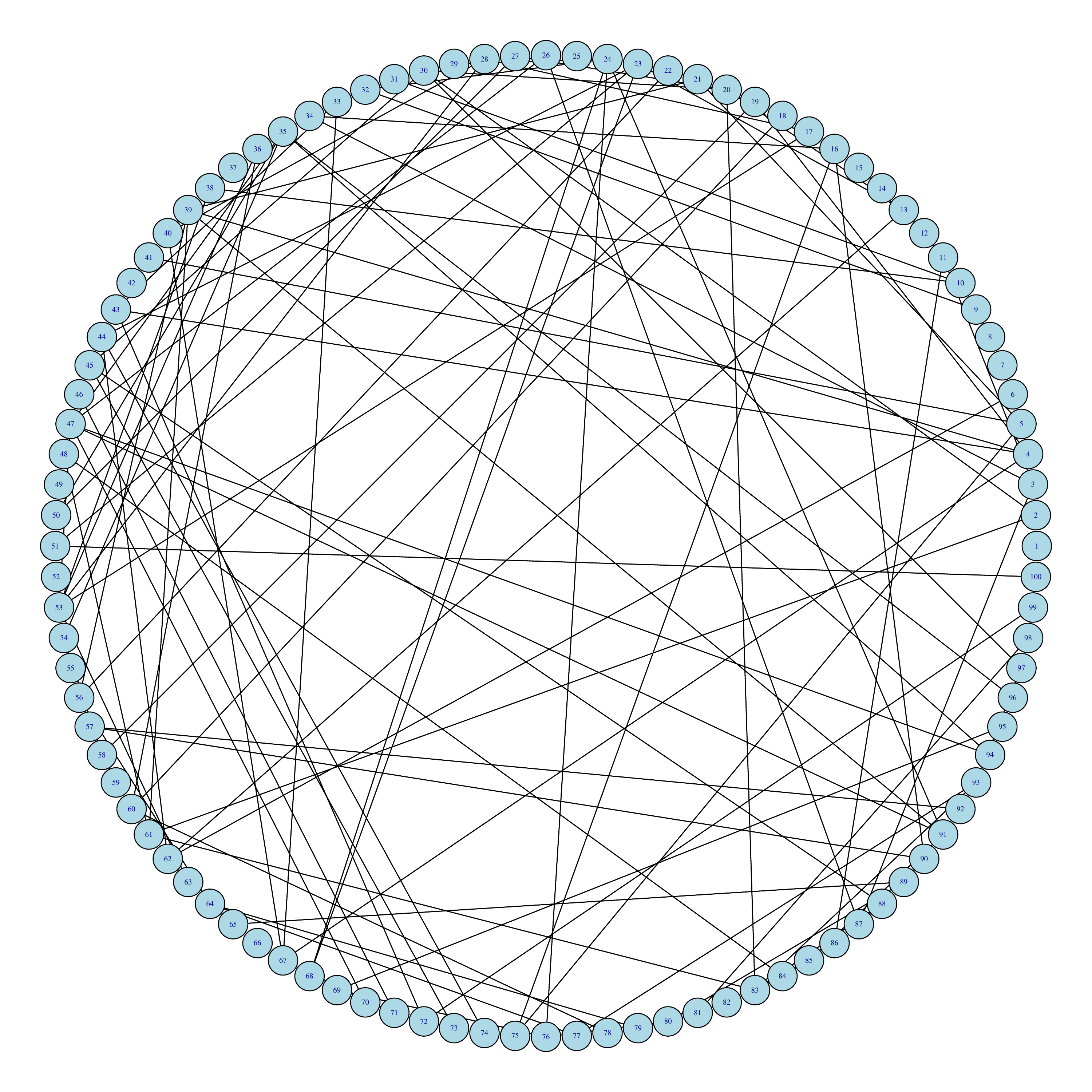}
         \caption{Ground Truth}
     \end{subfigure}
     \begin{subfigure}[b]{0.24\textwidth}
         \centering
         \includegraphics[width=\textwidth]{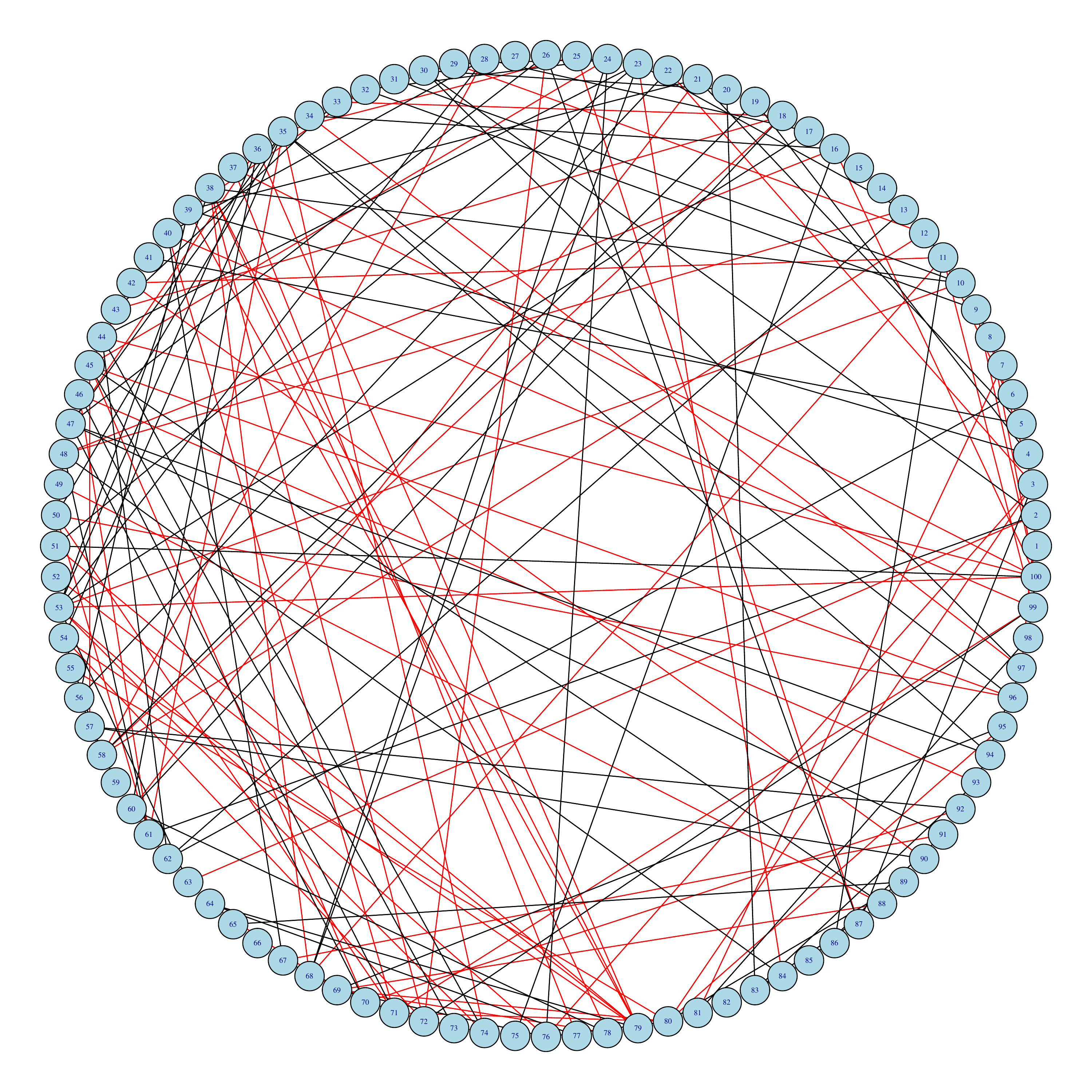}
         \caption{Graphical lasso tuned with cross-validation}
     \end{subfigure}
     \begin{subfigure}[b]{0.24\textwidth}
         \centering
         \includegraphics[width=\textwidth]{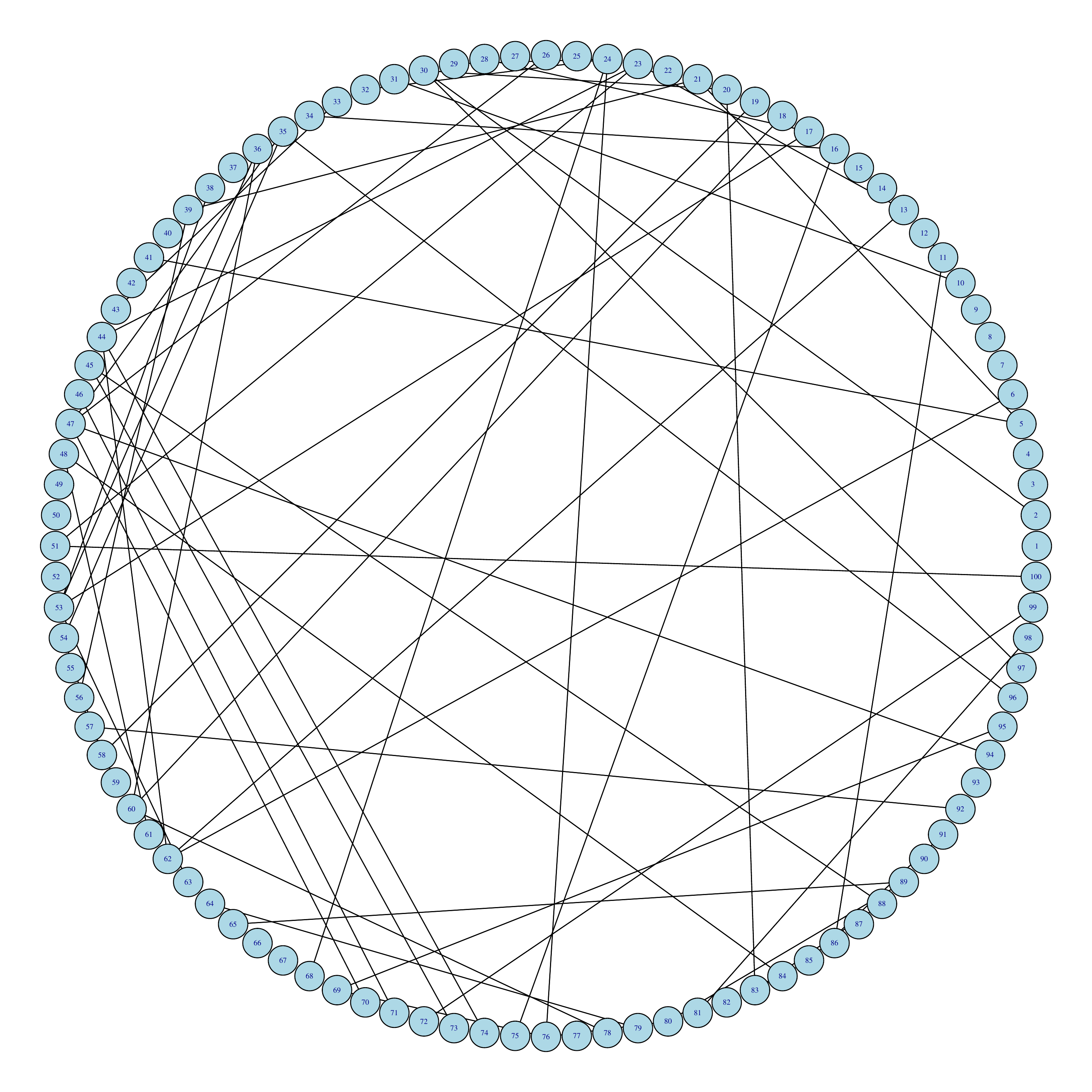}
         \caption{Graphical lasso tuned with RobSel ($\alpha=0.05$)}
     \end{subfigure}
     \begin{subfigure}[b]{0.24\textwidth}
         \centering
         \includegraphics[width=\textwidth]{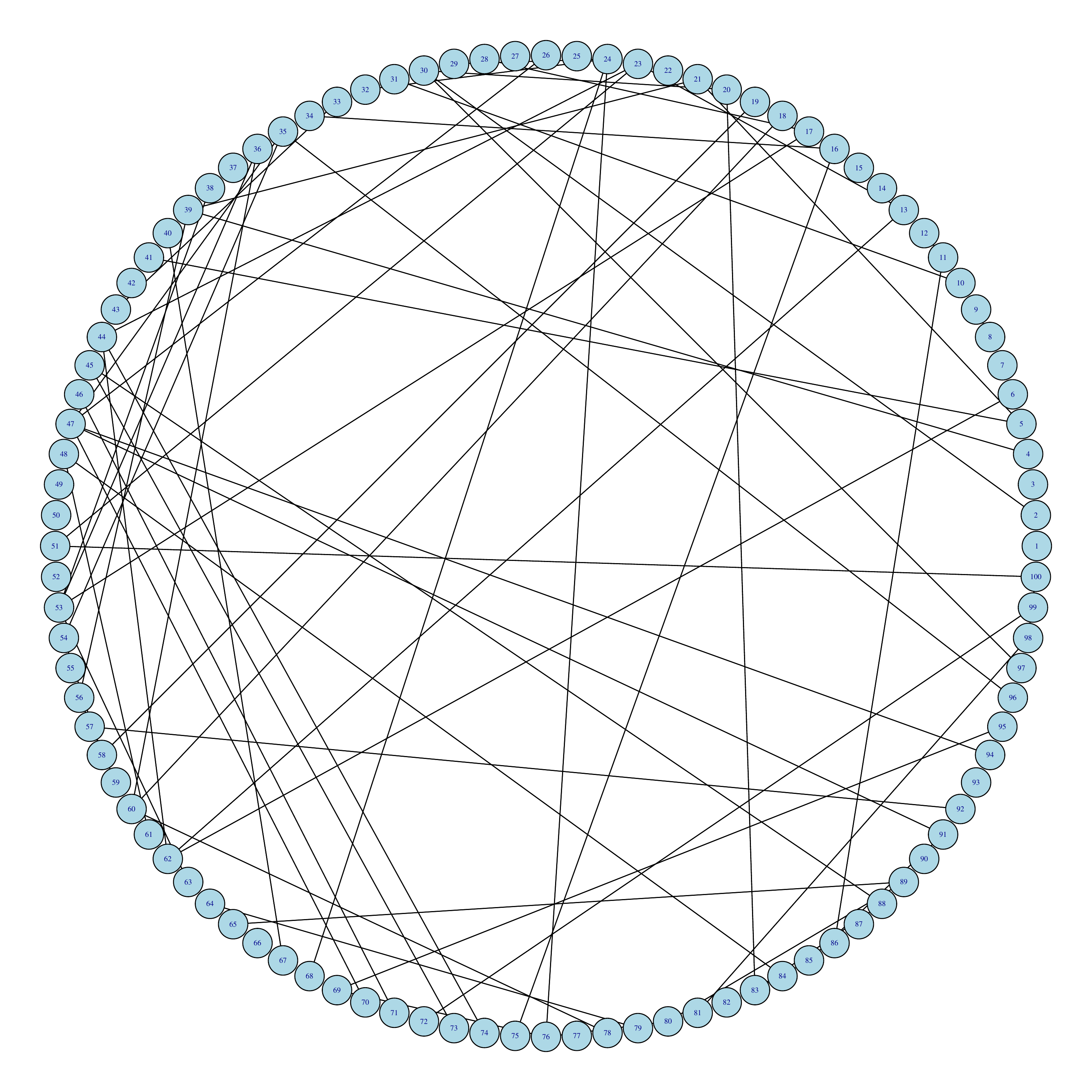}
         \caption{Holm-based multiple testing ($\alpha=0.05$)}
     \end{subfigure}
        \caption{True and three estimated graphs from a dataset with $n=3200$. Red edges denote False Positive edges. }
        \label{fig:estimated_graph}
\end{figure}

\subsection{Application to gene regulatory network reconstruction}\label{sec:gene}

Here, we infer gene regulatory networks from real datasets provided for the DREAM5 transcriptional network inference challenge \citep{Marbach2012}. We reconstructed the networks of interactions among transcription factors (TF). TF-encoding genes usually act as hub-genes with large numbers of interactions with other genes (\cite{JMLR:v15:tan14b}). Thus, identifying interactions between TFs may help researchers better understand the relationships between different groups of genes. The \emph{in silico} dataset contains $d=195$ transcription factors on $n=805$ arrays. The \emph{Escherichia coli} (E. coli) dataset contains $d=334$ transcription factors on $n=805$ arrays. The \emph{Saccharomyces cerevisiae} (S. cerevisiae) dataset contains $d=333$ transcription factors on $n=536$ arrays. To evaluate the inferred networks, we validated the edges in estimated graphical models against experimentally validated interactions given in \cite{Marbach2012}.

Graphical models were constructed using graphical lasso tuned with three different regularization parameter selection approaches as well as the using the Holm-corrected testing method described in Section \ref{sec2}. The regularization parameter tuning approaches we considered were as follows. The first is \emph{Robust Selection (RobSel)}, with $B=200$ sets of bootstrap samples. The second is \emph{$5$-fold cross-validation (CV)} procedure, where the performance on the validation set is the evaluation of the graphical loss function under the empirical measure of the samples on the training set. The third is extended Bayesian information criterion (EBIC) proposed in~\citep{RFBarber2010}. CV and EBIC are evaluated on the same grid of $\lambda$, which are ten logarithmically spaced values in the interval $(0.05s_{\max},s_{\max}]$ with $s_{\max}$ being the minimal value of regularization that gives an empty graph: i.e., setting $\lambda=s_{\max}$ for graphical lasso returning a diagonal matrix $\Omega$. Note that increasing the number of $\lambda$ values on the grid increases computational time. 

Because DRO framework minimizes worst case expected loss, specifying a small error tolerance $\alpha$ for RobSel often results in a graph with very few edges being estimated especially when analyzing a real dataset. In practice, a larger $\alpha$ might be beneficial in order to estimate graphs with more edges. Note, however, that setting a $\lambda$ corresponding to a large $\alpha$ when using graphical lasso would still return a very sparse graph. In our analyses, RobSel was specified with $\alpha=0.9$, EBIC with parameter $\gamma = 0.5$, and 5-fold for cross-validation. EBIC criterion has the following form:
\begin{align}
    \text{EBIC}_\gamma(E) = -2\mathcal{L}(\hat{\Omega}(E)) + \abs{E}\log n + \gamma4\abs{E}\log d, 
\end{align}
where $E$ is the edge set of a candidate graph implied by $\hat{\Omega}$, and $\mathcal{L}(\hat{\Omega}(E))$ denotes the maximized log-likelihood function of the associated model.

Table \ref{table:dream} show the number of edges in the estimated graph, number of validated edges \citep[interactions found in][]{Marbach2012}, the ratio of validated edge counts to total edge counts, and the wall clock times. In our results, an estimated edge (i.e. gene interaction) is a true positive if it is experimentally validated interaction in the database, i.e. in \cite{Marbach2012}. We can see that for all three data sets, RobSel appears to be faster than EBIC and CV with similar discovery ratios. Between E. coli and S. cerevisiae data sets, computational time for RobSel decreases when sample size decreases, but computational times of both EBIC and CV increase. Even though we used RobSel with $\alpha=0.9$ to get a denser graph, the estimated graph by RobSel are still much sparser than EBIC and CV.
\begin{table}[Ht!]
\centering
\begin{tabular}{ |c|c|c|c|c|c| } 
\hline
Dataset & Method & Estimated edges & Validated edges & Validated Edge Proportion & Time(s) \\
\hline
\multirow{4}{7em}{In silico} & Holm & 289 & 63 & {0.2184} & {0.088} \\ 
& RobSel & 693 & 89 & 0.1284 & 0.467 \\ 
& EBIC & 1237 & 108  & 0.0873& 1.566 \\ 
& CV & 7241 & {168}  & 0.0232& 8.611 \\ 
\hline
\multirow{4}{7em}{E. coli} & Holm & 269 & 14 & {0.0520} & {0.166} \\ 
& RobSel &  3479 & 22 & 0.0063 & 3.355 \\ 
& EBIC & 6599 & 37  & 0.0056& 10.46 \\ 
& CV & 10770 & {43}  & 0.0040& 52.92 \\ 
\hline
\multirow{4}{7em}{S. cerevisiae} & Holm & 56 & 3 & {0.0536} & {0.149} \\ 
& RobSel & 4259 & 46 & 0.0108 & 2.728 \\ 
& EBIC & 7731 & 70  & 0.0091&  17.80 \\ 
& CV & 11367 & {93}  & 0.0082& 85.64 \\ 
\hline
\end{tabular}
\caption{Graph recovery results and computational times in seconds from the DREAM5 datasets for three methods, Holm's testing procedure with $\alpha = 0.9$, RobSel with $\alpha = 0.9$, extended BIC (EBIC) with $\gamma = 0.5$, and 5-fold cross-validation (CV).}
\label{table:dream}
\end{table}

\section{Discussion}\label{discussion}
We made a theoretical connection between significant level $\alpha$ from RobSel and family-wise error rate of estimating any false positive edges when RobSel is used to tune graphical lasso. Furthermore, the asymptotic FWER control property is tested in finite sample using simulation experiments. The similarity between Holm-testing method and RobSel tuned graphical lasso solutions when using the same significance level $\alpha$ give users practical insight about the behavior of graphical lasso: graphical lasso regularization can be chosen according to a user specified FWER level. 


\section*{Acknowledgments}
This is acknowledgment. Provide text here. This is acknowledgment text. Provide text here. This is acknowledgment text. 

\subsection*{Author contributions}

This is an author contribution text. 

\subsection*{Financial disclosure}

None reported.

\bibliography{main}%


\end{document}